\documentclass[10pt, twocolumn]{IEEEtran}
\usepackage{hyperref}
\usepackage{amsfonts}
\usepackage{textcomp}
\usepackage{stfloats}
\usepackage{url}
\usepackage{graphicx}
\usepackage{cite}
\usepackage{subcaption}
\usepackage{dirtytalk}
\usepackage{multirow}

\usepackage{acronym}

\newacro{gd}[GD]{gradient decent}
\newacro{d-gd}[D-GD]{distributed gradient decent}
\newacro{sgd}[SGD]{stochastic gradient decent}

\newacro{fl}[FL]{federated learning}
\newacro{ml}[ML]{machine learning}

\newacro{ps}[PS]{parameter server}
\newacro{mmse}[MMSE]{minimum mean square error}
\newacro{se}[SE]{square error}

\newacro{ls}[LS]{least square}
\newacro{lse}[LSE]{least square error}

\newacro{mac}[MAC]{multiple access channel}

\newacro{urc}[URC]{unbiased random compressor}
\newacro{brc}[BRC]{biased random compressor}
\newacro{upc}[UPC]{unbiased predictive compressor}

\usepackage{amssymb,amsmath,amsfonts}
\usepackage{amsthm,thmtools}
\usepackage{bm,bbm}
\usepackage{mathrsfs}
\usepackage{mathtools}
\usepackage[ruled, vlined, linesnumbered]{algorithm2e}
\let\oldnl\nl
\newcommand{\nonl}{\renewcommand{\nl}{\let\nl\oldnl}}

\SetCommentSty{mycommfont}

\usepackage{suffix}

\usepackage{array}

\PassOptionsToPackage{table}{xcolor}
\PassOptionsToPackage{usenames}{xcolor}
\PassOptionsToPackage{dvipsnames}{xcolor}

\usepackage[framemethod=TikZ]{mdframed}

\newtheorem{theorem}{Theorem}[section]

\newtheorem{lemma}[theorem]{Lemma}

\newtheorem*{theorem*}{Theorem}
\newtheorem*{corollary*}{Corollary}
\newtheorem*{lemma*}{Lemma}

\newtheorem{definition}{Definition}[section]
\newtheorem*{definition*}{Definition}

\newtheorem{assumption}{Assumption}[section]
\newtheorem*{assumption*}{Assumption}

\newtheorem{remark}{Remark}[section]
\newtheorem*{remark*}{Remark}

\declaretheoremstyle[
	headfont=\bfseries,
	notefont=\mdseries, notebraces={(}{)},
	bodyfont=\normalfont,
	headpunct={\textnormal{:}},
	postheadspace=\newline,
	spacebelow=\parsep,
	spaceabove=\parsep,
	mdframed={
		backgroundcolor=red!02,
		linecolor=red!80,
		linewidth=0.3mm,
		innertopmargin=6pt,
		innerbottommargin=6pt,
		roundcorner=10pt, }
]{red style name}

\declaretheoremstyle[
	headfont=\bfseries,
	notefont=\mdseries, notebraces={(}{)},
	bodyfont=\normalfont,
	headpunct={\textnormal{:}},
	postheadspace=\newline,
	spacebelow=\parsep,
	spaceabove=\parsep,
	mdframed={
		backgroundcolor=blue!02,
		linecolor=blue!80,
		linewidth=0.3mm,
		innertopmargin=6pt,
		innerbottommargin=6pt,
		roundcorner=10pt, }
]{blue style name}

\declaretheoremstyle[shaded={rulecolor=Lavender,
	rulewidth=2pt, bgcolor={rgb}{1,1,1}},spaceabove=6pt, spacebelow=6pt,
headfont=\normalfont\bfseries,
notefont=\mdseries, notebraces={(}{)},
bodyfont=\normalfont,
postheadspace=1em,
]{aaaaa}

\usepackage{etoolbox}
\newcommand{\subscriptifnonempty}[1]{\ifstrempty{#1}{}{_{#1}}}

\makeatletter
\newcommand{\@somedarnpadding}{\mathchoice
  {\mkern-4mu}
  {\mkern-4mu}
  {}
  {}
}
\newcommand{\@somedarnpaddingstar}{\mathchoice
  {}
  {}
  {}
  {}
}
\newcommand{\@somedarnpaddingwithsubscript}{\mathchoice
  {\mkern-4mu}
  {\mkern-4mu}
  {}
  {}
}
\newcommand{\@somedarnpaddingwithsubscriptstar}{\mathchoice
  {}
  {}
  {}
  {}
}
\makeatother

\makeatletter
\newcommand{\cbigoplus}{\DOTSB\cbigoplus@\slimits@}
\newcommand{\cbigoplus@}{\mathop{\widehat{\bigoplus}}}
\makeatother

\makeatletter
\newcommand{\cbigotimes}{\DOTSB\cbigotimes@\slimits@}
\newcommand{\cbigotimes@}{\mathop{\widehat{\bigotimes}}}
\makeatother

\makeatletter
\DeclareMathOperator{\E}{\mathbb{E}}
\newcommand\EX[2][]{{\E\subscriptifnonempty{#1}\ifstrempty{#1}{\@somedarnpadding}{\@somedarnpaddingwithsubscript}\left\{ #2 \right\}}}
\WithSuffix\newcommand\EX*[2][]{{\E\subscriptifnonempty{#1}\ifstrempty{#1}{\@somedarnpaddingstar}{\@somedarnpaddingwithsubscriptstar}\{ #2 \}}}
\makeatother

\makeatletter
\DeclareMathOperator{\Var}{\text{Var}} 
\DeclareMathOperator{\Cov}{\text{Cov}} 
\newcommand{\VarX}[2][]{{\Var\subscriptifnonempty{#1}\ifstrempty{#1}{\@somedarnpadding}{\@somedarnpaddingwithsubscript}\left\{ #2 \right\}}}
\newcommand{\CovX}[3][]{{\Cov\subscriptifnonempty{#1}\ifstrempty{#1}{\@somedarnpaddingstar}{\@somedarnpaddingwithsubscriptstar}\left\{ #2,#3 \right\}}}
\makeatother

\newcommand\abs[1]{{\left\lvert#1\right\rvert}}
\WithSuffix\newcommand\abs*[1]{{\lvert#1\rvert}}

\newcommand\scalarprod[2]{\langle #1,#2 \rangle}
\WithSuffix\newcommand\scalarprod*[2]{\left\langle #1,#2\right\rangle}

\newcommand{\set}[1]{{\left\{#1\right\}}}
\WithSuffix\newcommand\set*[1]{{\{#1\}}}

\DeclareMathOperator{\Entropy}{\textit{H}}
\newcommand{\HX}[2][]{{\Entropy\subscriptifnonempty{#1}\left\{ #2 \right\}}}
\WithSuffix\newcommand\HX*[2][]{{\Entropy\subscriptifnonempty{#1}\{ #2 \}}}

\DeclareMathOperator*{\argmin}{\arg\!\min}

\newcommand{\minimize}{\mathop{\rm minimize}\limits}

\newcommand{\vect}[1]{\bm{#1}}

\DeclareMathOperator{\rank}{rank}
\newcommand{\transpose}{^{\mkern-1.5mu\mathsf{T}}}

\newcommand{\zerovec}[1]{\mathbf{0}\subscriptifnonempty{#1}}
\DeclareMathAlphabet{\mathbbx} {U}{BOONDOX-ds}{b}{n}

\newcommand\norm[2][]{\left\lVert #2 \right\rVert\subscriptifnonempty{#1}}
\WithSuffix\newcommand\norm*[2][]{\Vert #2 \Vert\subscriptifnonempty{#1}}

\newcommand{\mnorm}[2][]{
		\left\vert\kern-0.25ex\left\vert\kern-0.25ex\left\vert
		#2
		\right\vert\kern-0.25ex\right\vert\kern-0.25ex\right\vert\subscriptifnonempty{#1}
}

\newcommand{\bernoulli}[1]{{\text{Bernoulli}\left( #1 \right)}}

\makeatletter
\newcommand{\@letterdotthing}[1]{\ifstrempty{#1}{\cdot}{#1}}
\newcommand{\@letterthing}[2]{\ensuremath{#1\@somedarnpadding\left(\@letterdotthing{#2}\right)}}
\newcommand{\@letterthingstar}[2]{\ensuremath{#1(\@letterdotthing{#2})}}

\newcommand{\BB}{\mathcal{B}}
\newcommand{\BBX}[1]{\@letterthing{\BB}{#1}}
\WithSuffix\newcommand\BBX*[1]{\@letterthingstar{\BB}{#1}}

\newcommand{\CC}{\mathcal{C}}
\newcommand{\CCX}[1]{\@letterthing{\CC}{#1}}
\WithSuffix\newcommand\CCX*[1]{\@letterthingstar{\CC}{#1}}

\newcommand{\KK}{\mathcal{K}}
\newcommand{\KKX}[1]{\@letterthing{\KK}{#1}}
\WithSuffix\newcommand\KKX*[1]{\@letterthingstar{\KK}{#1}}

\newcommand{\PP}{\mathcal{P}}
\newcommand{\PPX}[1]{\@letterthing{\PP}{#1}}
\WithSuffix\newcommand\PPX*[1]{\@letterthingstar{\PP}{#1}}

\newcommand{\QQ}{\mathcal{Q}}
\newcommand{\QQX}[1]{\@letterthing{\QQ}{#1}}
\WithSuffix\newcommand\QQX*[1]{\@letterthingstar{\QQ}{#1}}
\makeatother

\newcommand{\reals}{\mathbb{R}}

\newcommand{\naturals}{\mathbb{N}}
\newcommand{\integers}{\mathbb{Z}}

\def\shrug{\texttt{\raisebox{0.75em}{\char`\_}\char`\\\char`\_\kern-0.5ex(\kern-0.25ex\raisebox{0.25ex}{\rotatebox{45}{\raisebox{-.75ex}"\kern-1.5ex\rotatebox{-90})}}\kern-0.5ex)\kern-0.5ex\char`\_/\raisebox{0.75em}{\char`\_}}}

\usepackage{tabularx}
\usepackage{graphicx}
\PassOptionsToPackage{table}{xcolor}
\PassOptionsToPackage{usenames}{xcolor}
\PassOptionsToPackage{dvipsnames}{xcolor}

\usepackage{pgf}
\usepackage{pgfplots}
\usepackage{tikz}
\usepackage{tikz-3dplot}

\usetikzlibrary{positioning}
\usetikzlibrary{decorations.pathreplacing}
\usetikzlibrary{plotmarks}
\usetikzlibrary{patterns}
\usetikzlibrary{calc}
\usetikzlibrary{spy}
\usetikzlibrary{fit}
\usetikzlibrary{backgrounds}
\usepgfplotslibrary{groupplots}
\pgfplotsset{
	compat=1.18,
}
\usepackage{pgfplotstable}

\pgfplotsset{filter discard warning=false}
\pgfplotsset{every axis plot/.append style={line width = 2pt},axis background/.style={fill=white}}

\usetikzlibrary{external}
\usepgfplotslibrary{external}

\makeatletter
\DeclareRobustCommand{\rvdots}{%
	\vbox{
		\baselineskip5\p@\lineskiplimit\z@
		\kern-\p@
		\hbox{.}\hbox{.}\hbox{.}
}}
\makeatother

\tikzset{syscircle/.style={
		draw,
		circle,
		minimum size=0.6cm,
		fill=blue!05},
}
\tikzset{sysblock/.style={
		draw,
		minimum width=1.1cm,
		minimum height=1.0cm,
		fill=red!05},
}
\tikzset{sysline/.style={
		->,
		thick
	},
}
\tikzset{syslinerev/.style={
		sysline,
		<-,
	},
}
\tikzset{sysjunction/.style={
		circle,
		minimum size=1.5mm,
		inner sep=0,
		fill=black
	},
}
\tikzset{sysfitblock/.style={
		inner sep=5mm,
		fill=cyan!05,
		densely dashed,
		thick,
		draw=black,
		rounded corners=8pt,
	},
}
\tikzset
{
	port/.style     = {inner sep=0pt, font=\tiny},
	cross/.style =
	{%
		path picture=%
		{
			\draw
			(path picture bounding box.north west) --
			(path picture bounding box.south east)
			(path picture bounding box.south west) --
			(path picture bounding box.north east)
			;
		}
	},
	syssum/.style n args = {4}%
	{%
		syscircle, node distance = 2cm, minimum size=9mm, alias=sum,
		append after command=%
		{%
			node at (sum.north) [port, below=1pt] {$#1$}
			node at (sum.west) [port, right=1pt] {$#2$}
			node at (sum.south) [port, above=1pt] {$#3$}
			node at (sum.east) [port, left=1pt] {$#4$}
			node at (sum) [port, font=\small] {$\sum$}
		},
	},
}

\tikzset{wireless/.pic=
	{
		\node {\small Placeholder};
	}
}

\tikzset{sysswitch/.pic=
	{
		\node[sysjunction,alias=tttt] at (-4mm,0) {};
		\node[sysjunction] at (4mm,0) {};
		\draw[sysline,-] (tttt) to (3.5mm,3mm);
		\path (-4mm,0) to (3.5mm,-3mm);
	}
}

\tikzset{
	pics/portal/.style args={#1/#2}{
		code={
			\path[fill=#2,line width=0.0cm] (0.01, -1.39).. controls (0.05, -1.37)
			and (0.1, -1.34) .. (0.07, -1.35).. controls (0.03, -1.36) and (-0.03, -1.37)
			.. (-0.09, -1.36).. controls (-0.27, -1.35) and (-0.42, -1.22) .. (-0.53,
			-0.99).. controls (-0.59, -0.85) and (-0.64, -0.68) .. (-0.63, -0.63)..
			controls (-0.63, -0.62) and (-0.62, -0.62) .. (-0.62, -0.62).. controls
			(-0.62, -0.62) and (-0.6, -0.67) .. (-0.58, -0.72).. controls (-0.47, -1.02)
			and (-0.28, -1.21) .. (-0.08, -1.21).. controls (0.08, -1.21) and (0.29,
			-1.07) .. (0.43, -0.86).. controls (0.59, -0.63) and (0.7, -0.34) .. (0.75,
			-0.06).. controls (0.78, 0.1) and (0.79, 0.17) .. (0.79, 0.34).. controls
			(0.79, 0.55) and (0.77, 0.71) .. (0.72, 0.89) -- (0.7, 0.96) -- (0.71, 0.75)..
			controls (0.71, 0.44) and (0.7, 0.28) .. (0.64, 0.13).. controls (0.62, 0.07)
			and (0.58, 0.01) .. (0.58, 0.02).. controls (0.57, 0.03) and (0.58, 0.08) ..
			(0.59, 0.15).. controls (0.61, 0.3) and (0.62, 0.51) .. (0.61, 0.65)..
			controls (0.59, 1.02) and (0.47, 1.25) .. (0.26, 1.31).. controls (0.15, 1.35)
			and (0.08, 1.34) .. (-0.04, 1.28).. controls (-0.23, 1.18) and (-0.38, 1.02)
			.. (-0.5, 0.78).. controls (-0.61, 0.56) and (-0.69, 0.27) .. (-0.75, -0.13)..
			controls (-0.78, -0.31) and (-0.78, -0.55) .. (-0.76, -0.7).. controls
			(-0.73, -0.92) and (-0.66, -1.09) .. (-0.55, -1.21).. controls (-0.48, -1.3)
			and (-0.36, -1.37) .. (-0.25, -1.4).. controls (-0.16, -1.42) and (-0.05,
			-1.42) .. (0.01, -1.39) -- cycle;

			\path[fill=#1,line width=0.0cm] (-0.3, -1.59).. controls (-0.44, -1.57)
			and (-0.53, -1.52) .. (-0.64, -1.42).. controls (-0.86, -1.19) and (-0.98,
			-0.88) .. (-1.0, -0.46).. controls (-1.01, -0.25) and (-0.97, 0.15) .. (-0.91,
			0.42).. controls (-0.84, 0.78) and (-0.69, 1.11) .. (-0.54, 1.29).. controls
			(-0.5, 1.33) and (-0.45, 1.38) .. (-0.44, 1.37).. controls (-0.44, 1.37) and
			(-0.45, 1.34) .. (-0.47, 1.3).. controls (-0.58, 1.13) and (-0.72, 0.82) ..
			(-0.72, 0.77).. controls (-0.72, 0.75) and (-0.72, 0.75) .. (-0.71, 0.77)..
			controls (-0.55, 1.07) and (-0.36, 1.32) .. (-0.17, 1.45).. controls (-0.1,
			1.5) and (0.02, 1.57) .. (0.09, 1.58).. controls (0.32, 1.64) and (0.59, 1.54)
			.. (0.76, 1.34).. controls (0.88, 1.2) and (0.96, 1.0) .. (0.99, 0.76)..
			controls (1.0, 0.67) and (1.0, 0.34) .. (0.99, 0.23).. controls (0.96, -0.05)
			and (0.84, -0.47) .. (0.71, -0.78).. controls (0.5, -1.25) and (0.24, -1.53)
			.. (-0.06, -1.59).. controls (-0.12, -1.6) and (-0.24, -1.6) .. (-0.3, -1.59)
			-- cycle(0.01, -1.39).. controls (0.05, -1.37) and (0.1, -1.34) .. (0.07,
			-1.35).. controls (0.03, -1.36) and (-0.03, -1.37) .. (-0.09, -1.36)..
			controls (-0.27, -1.35) and (-0.42, -1.22) .. (-0.53, -0.99).. controls
			(-0.59, -0.85) and (-0.64, -0.68) .. (-0.63, -0.63).. controls (-0.63, -0.62)
			and (-0.62, -0.62) .. (-0.62, -0.62).. controls (-0.62, -0.62) and (-0.6,
			-0.67) .. (-0.58, -0.72).. controls (-0.47, -1.02) and (-0.28, -1.21) ..
			(-0.08, -1.21).. controls (0.08, -1.21) and (0.29, -1.07) .. (0.43, -0.86)..
			controls (0.59, -0.63) and (0.7, -0.34) .. (0.75, -0.06).. controls (0.78,
			0.1) and (0.79, 0.17) .. (0.79, 0.34).. controls (0.79, 0.55) and (0.77, 0.71)
			.. (0.72, 0.89) -- (0.7, 0.96) -- (0.71, 0.75).. controls (0.71, 0.44) and
			(0.7, 0.28) .. (0.64, 0.13).. controls (0.62, 0.07) and (0.58, 0.01) .. (0.58,
			0.02).. controls (0.57, 0.03) and (0.58, 0.08) .. (0.59, 0.15).. controls
			(0.61, 0.3) and (0.62, 0.51) .. (0.61, 0.65).. controls (0.59, 1.02) and
			(0.47, 1.25) .. (0.26, 1.31).. controls (0.15, 1.35) and (0.08, 1.34) ..
			(-0.04, 1.28).. controls (-0.23, 1.18) and (-0.38, 1.02) .. (-0.5, 0.78)..
			controls (-0.61, 0.56) and (-0.69, 0.27) .. (-0.75, -0.13).. controls (-0.78,
			-0.31) and (-0.78, -0.55) .. (-0.76, -0.7).. controls (-0.73, -0.92) and
			(-0.66, -1.09) .. (-0.55, -1.21).. controls (-0.48, -1.3) and (-0.36, -1.37)
			.. (-0.25, -1.4).. controls (-0.16, -1.42) and (-0.05, -1.42) .. (0.01, -1.39)
			-- cycle;

		}
	}
}

\pgfplotscreateplotcyclelist{cyclelist}{
	blue, solid, every mark/.append style={fill=blue}, mark=*,mark repeat={10},mark size=1pt\\
	red,  solid, every mark/.append style={fill=red}, mark=triangle*,mark repeat={10},mark size=1pt\\
	cyan, dashdotdotted \\
	red, densely dashed,every mark/.append style={fill=red}, mark=diamond*,mark repeat={10},mark size=1pt, mark phase = {5},\\
	darkgray, densely dotted\\
	orange, dashdotted\\
}

\pgfplotsset{singlecolumnsize/.style={
		width = 8.6cm,
		height = 4.8 cm,
	},
}

\let\hat\widehat
\let\tilde\widetilde

\usepackage[noabbrev,capitalise]{cleveref}

\crefname{assumption}{Assumption}{Assumptions}
\Crefname{assumption}{Assumption}{Assumptions}

\hypersetup{
	colorlinks,
	linkcolor={black},
	citecolor={blue!80!black},
	urlcolor={blue!80!black}
}

\title{Temporal Predictive Coding for Gradient Compression in Distributed Learning}
\author{Adrian Edin$^*$, Zheng Chen$^*$, Michel Kieffer$^\dag$, and Mikael Johansson$^\ddag$\\[1mm]
$^*$ Department of Electrical Engineering, Linköping University, Sweden \\Email: \{adrian.edin, zheng.chen\}@liu.se\\
$^\dag$ Université Paris-Saclay, CentraleSupelec, CNRS, Laboratoire des Signaux et Systèmes, Gif-sur-Yvette, France \\ Email: michel.kieffer@l2s.centralesupelec.fr\\
$^\ddag$ School of Electrical Engineering and Computer Science, KTH, Stockholm, Sweden  \\Email: mikaelj@kth.se
}

\newcommand{\smemory}{\ensuremath{\mathcal{M}}}

\usepackage{suffix}
\newcommand{\ts}[1]{^{(#1)}}
\WithSuffix\newcommand\ts*[1]{{\ts{#1}}}

\begin{document}

\maketitle

\begin{abstract}
This paper proposes a prediction-based gradient compression method for distributed learning with event-triggered communication. Our goal is to reduce the amount of information transmitted from the distributed agents
to the parameter server by exploiting temporal correlation in the local gradients. We use a linear predictor that \textit{combines past gradients to form a prediction of the current gradient}, with coefficients that are optimized by solving a least-square problem.
In each iteration, every agent transmits the predictor coefficients to the server such that the predicted local gradient can be computed. The difference between the true local gradient and the predicted one, termed the \textit{prediction residual, is only transmitted when its norm is above some threshold.} When this additional communication step is omitted, the server uses the prediction as the estimated gradient. This proposed design shows notable performance gains compared to existing methods in the literature, achieving convergence with reduced communication costs.
\end{abstract}


\acresetall

\section{Introduction}

In the era of the Internet of Things, numerous devices collect an ever-increasing amount of data, which can be used to train \ac{ml} models, enabling intelligent decision making. Uploading large amounts of training data to a central server generates a high communication load and raises privacy concerns. Motivated by this, distributed and collaborative \ac{ml} has emerged as an efficient solution for training a common model across multiple agents without sharing raw data \cite{verbraeken2020survey}.

Distributed learning involves an iterative process of on-device local training and server-assisted model aggregation, which requires frequent communication between the agents and the \ac{ps}. The communication bottleneck is one of the most significant challenges in distributed learning, especially when the communication links have limited capacity due to resource constraints \cite{comm-effi-dl}. 
Alleviating this bottleneck is commonly done by compressing the model updates, using quantization and/or sparsification \cite{konecny2016improvedcommeff, khirirat2018gradientcompression, alistarh2018convergencesparse,shokri2015privacy}.
Another common approach is to reduce the number or frequency of communication links, either through scheduling \cite{gurbuzbalaban2017convergenceiag, hu2023scheduling}, or using an event-triggered design \cite{zhong2010distributedevent}.
Using low-rank approximation of the gradients is another efficient method for achieving aggressive compression
\cite{wang2018atomo, vogels2019powersgd}.

During the training process, the parameters of the local models evolve slowly over time, which implies potential correlation between local models or gradients obtained at consecutive iterations.
Some recent studies have considered using predictive coding techniques for model parameter prediction and compression \cite{yue2022predictive, song2024resfed}. The key idea is to use past local model parameters to obtain a predicted version of the current model parameters. Then, the prediction residual, computed as the difference between the true and predicted values of the model parameters, is communicated between each agent and the \ac{ps}.
Similarly, \cite{mishchenko2019graddiff} proposes a framework where the gradient difference is transmitted in every iteration. This can be viewed as a special case that uses the previous gradient as a prediction of the current gradient. This prediction method was further analyzed in \cite{richtarik2021ef21} for a larger class of compressors. As an additional step, communication frequency can be further reduced by avoiding transmission when the gradient difference is small
\cite{chen2018lag, sun2022laq}.

Note that all the aforementioned studies apply some heuristic predictor designs where the prediction is calculated without accounting for the current model parameters or gradient.
In this work, we aim to further improve the performance by considering a \ac{ls} estimator and an event-triggered design to limit the amount of communicated data. First, the predicted gradient is obtained by optimally combining the past gradients, ensuring that the prediction residual is never larger than the gradient itself. Second, the residual is only transmitted when its norm exceeds a certain threshold, avoiding additional communications when the prediction is sufficiently accurate. We provide theoretical analysis on the convergence performance of our proposed approach, and verify its effectiveness in improving communication efficiency by simulation results.

\subsubsection*{Notation} Scalars are represented by regular characters, vectors are in bold, and matrices are in upper case. All vectors are by default column vectors. $\zerovec{n}\in\reals^n$ is an $n$-dimensional vector of all zeros. $\norm{\vect{a}}$ is the 2-norm of vector $\vect{a}$. $\scalarprod{\vect{a}}{\vect{b}} = \vect{a}\transpose\vect{b}$ is the standard scalar product. $A\transpose$ is the transpose of $A$ and $A^\dagger$ is the pseudo-inverse of $A$.
The expected value of a random vector is $\EX{\vect{a}}$ and its variance is $\VarX{\vect{a}} = \EX*{(\vect{a}-\EX*{\vect{a}})\transpose (\vect{a}-\EX*{\vect{a}})}=\EX{\norm{\vect{a}^2}} - \norm{\EX{\vect{a}}}^2$. $[K]$ denotes the set of integers $\set{1,...,K}$. $\reals_+$ denotes the set of positive real numbers $\set{a:0<a\in\reals}$.

\section{System Model}

Many engineering problems, including distributed learning, can be formulated as solving a finite-sum optimization problem written as
\begin{align}\label{eq:sys:original problem}
	\minimize_{\vect{x}\in\reals^{d}} f(\vect{x}) = \sum_{k = 1}^Nf_k(\vect{x}),
\end{align}
where the objective $f(\vect{x})$ is written as the sum of some local (private) objective functions $f_k(\vect{x}):\reals^d\rightarrow\reals, \forall k\in[K]$.
When the objective function is differentiable, a classical first-order optimization method  for solving \eqref{eq:sys:original problem} is known as the \ac{gd}. It entails an iterative updating procedure
\begin{align}\label{eq:sys:gd step}
	\vect{x}\ts{t+1} = \vect{x}\ts{t} - \gamma\nabla f(\vect{x}\ts{t}) = \vect{x}\ts{t} - \gamma\vect{g}\ts{t}, t=0,1,\ldots,
\end{align}
where $\vect{x}\ts{0}$ is some initial point, $\vect{x}\ts{t}$ is the  parameter vector in iteration $t\geq 1$, $\nabla f(\vect{x}\ts{t})\in\reals^d$ is the gradient of the objective function evaluated at the point $\vect{x}\ts{t}$, and $\gamma$ is a step size. For several sets of functions, \textit{e.g.}, smooth and/or convex,  \ac{gd} has well known convergence guarantees \cite{nesterov2004introductory,bottou2018optimization,khirirat2018distributed}.

In the scenario of distributed learning, a network of $K$ agents aims to train a common learning model parameterized by $\vect{x}\in\reals^{d}$ using locally available training data samples. The objective of model training is essentially optimizing a given loss function $f(\vect{x})$, which can be written as the sum of local loss functions $f_k(\vect{x})$, $\forall k \in[K]$.
Typical choices of loss function for learning tasks include $\ell_2$ loss, cross-entropy, KL divergence, \emph{etc}.

We consider a master-worker setting, where a central \ac{ps} is responsible for aggregating and synchronizing the local models across different agents. The training procedure employs \ac{d-gd} where each iteration consists of four steps \cite{comm-effi-dl}:
\begin{enumerate}
	\item The \ac{ps} broadcasts the current parameter vector $\vect{x}\ts{t}$ to the agents;
	\item Using locally available training datasets, each agent computes its local gradient $\vect{g}_k\ts{t} = \nabla f_k(\vect{x}\ts{t})$, $\forall k \in [K]$;
	\item The $K$ agents transmit their local gradients to the PS;
	\item The \ac{ps} aggregates local gradients to evaluate the global gradient
	$\vect{g}\ts{t} = \sum_{k=1}^K\vect{g}_k\ts{t} = \sum_{k=1}^K\nabla f_k(\vect{x}\ts{t})$, which is then used in \eqref{eq:sys:gd step} to update the global model.
\end{enumerate}

When the number of agents and/or the dimension of the parameter vector is very large, the amount of information that needs to be communicated becomes substantial.
This issue is particularly prominent when the communication links are rate-constrained.
Thus, compressing the local gradients before transmission becomes necessary.

Gradient compression can be accomplished in many ways, such as quantization and masking/sparsification.
One emerging method is to predict the current gradient based on information from the past gradients and only transmit the prediction difference/residual. Prediction-based compression exploits the correlation between consecutive gradients to reduce the information to be transmitted. The correlation could be due to inherent properties of the objective function, or as a consequence of regularization or momentum-based techniques \cite{bottou2018optimization,lin2018deepgradient,karimireddy2020scaffold}.

Existing works on prediction-based gradient compression for distributed learning have considered heuristic predictor designs, which cannot provide guaranteed performance improvement as compared to the non-predictive case.
In this work, we use an \ac{ls}-optimal linear predictor with compression performance guarantee and further reduce the amount of communicated data by adopting an event-triggered design.

\section{Gradient Compression with Predictive Coding}\label{sec:predictive coding}

In each iteration $t$ of the training process, agent $k$ performs gradient prediction using a \textit{memory} $\smemory_k\ts{t}$ containing past (possibly imperfect) gradient information, and \textit{predictor coefficients} $\vect{a}_k\ts{t}$. The predicted gradient is written as
\begin{equation}\label{eq:pred:predictor}
	\hat{\vect{g}}_k\ts{t} = p(\vect{a}_k\ts{t},\smemory_k\ts{t}),
\end{equation}
where $p(\cdot,\cdot)$ is the prediction function to be designed.
Here, the memory information is locally maintained and updated by each agent and by the \ac{ps} in a synchronized manner, without any required communication. The \ac{ps} maintains a copy of each agent's memory.
Given the current gradient and the predicted one, agent $k$ can compute the prediction error (residual) as
\begin{align}\label{eq:g tilde}
	\vect{e}_k\ts{t}=\vect{g}_k\ts{t}- \hat{\vect{g}}_k\ts{t}.
\end{align}
Let \CCX{} represent a certain compression scheme. Rather than transmitting $d$-dimensional vectors of full precision, each agent only transmits a few prediction coefficients $\vect{a}_k\ts{t}$ and a highly compressed prediction residual $\widetilde{\vect{e}}_k\ts{t}=\CC(\vect{e}_k\ts{t})$ to the server at each iteration.
Consequently, the original gradient cannot be perfectly reconstructed, and the \ac{ps} obtains the imperfect gradient estimate
\begin{align}
	\widetilde{\vect{g}}_k\ts{t} = \hat{\vect{g}}_k\ts{t} + \widetilde{\vect{e}}_k\ts{t}.
\end{align}
The imperfect gradient is used to update the memory $\smemory_k\ts{t+1}$ at agent $k$ and at the \ac{ps} before the next iteration.

Note that the prediction coefficients $\vect{a}_k\ts{t}$ typically have much smaller size than the full gradient, for example, it might be a single scalar parameter. Thus, we can allocate sufficient bits to represent these coefficients such that the distortion due to lossy compression of $\vect{a}_k$ is negligible.\footnote{The distortion measure may be for example the squared error distortion $\|\vect{a}_k\ts{t}-\widetilde{\vect{a}}_k\ts{t}\|^2$, where $\widetilde{\vect{a}}_k\ts{t}$ is the received compressed coefficients. In simulations, we use $16$, or $32$, bits for quantizing each prediction coefficient depending on the scenario.}

\subsection{Linear Prediction using Least Square Estimator}\label{sec:lienar predictor ls}
Our prediction and encoding design is inspired by predictive/differential coding methods that are commonly used for image compression. In this work, we consider a linear predictor that combines the past imperfect gradients to form the current prediction.

Assume that the memory contains the $s$ most recent imperfect gradients, \textit{i.e.}, $\smemory_k\ts{t} = \set{\widetilde{\vect{g}}_k\ts{t-i}}_{i = 1,...,s}$, and the predicted gradient is evaluated as
\begin{align}\label{eq:linear predictor}
	\hat{\vect{g}}_k\ts{t} =  \sum_{i = 1}^{s} a_{k,i}\ts{t}\widetilde{\vect{g}}_k\ts{t-i},
\end{align}
where $a_{k,i}\ts{t}$ is the predictor coefficient associated to the past imperfect gradient in iteration $t-i$.
The predictor in \eqref{eq:linear predictor} can be written in an equivalent matrix form as
\begin{align}\label{eq:linear predictor matrix form}
	\hat{\vect{g}}_k\ts{t}
	=
	\underbrace{\begin{pmatrix}
		\widetilde{\vect{g}}_k\ts{t-1} & \hdots & \widetilde{\vect{g}}_k\ts{t-s}
	\end{pmatrix}}_{G_k\ts{t}}
	\underbrace{\begin{pmatrix}
		a_{k,1}\ts{t}&
		\hdots&
		a_{k,s}\ts{t}
	\end{pmatrix}\transpose}_{\vect{a}_k\ts{t}},
\end{align}
where $G_k\ts{t}\in\reals^{d \times s}$ and $\vect{a}_k\ts{t}\in\reals^{s}$.

To find the optimal coefficients in the linear predictor, we use the \ac{ls} estimator, given as
\begin{equation}\label{eq:LSE solution}
	\vect{a}_k^*\ts*{t}
	=
	\argmin_{\vect{a}_k\ts{t}\in\reals^{s}} \norm{\vect{g}_k\ts{t} - G_k\ts{t}\vect{a}_k\ts{t}}^2,
\end{equation}
which has the known closed-form solution if $G_k\ts{t}$ has full column rank ($\rank G_k\ts{t} = s$)
\begin{equation}\label{eq:LSE solution:closed form}
\vect{a}_k^*\ts*{t}=(G_k\ts{t})^\dagger\vect{g}_k\ts{t},
\end{equation}
Here, we use the \ac{ls} predictor because it is simple and does not rely on any assumptions about the stationarity of the gradient distribution.

The geometrical interpretation of the \ac{ls} estimation is a projection of $\vect{g}_k\ts{t}$ onto the subspace $\mathcal{T}$ spanned by the column vectors of $G_k\ts{t}$, \textit{i.e.}, $\{\widetilde{\vect{g}}_k\ts{t-i}\}_{i=1,\ldots,s}$.
Consequently, the residual $\vect{e}_k\ts{t}$ is orthogonal to the prediction $\hat{\vect{g}}_k\ts{t}$.
This results in some favourable properties, for example, $\norm*{{\vect{e}_k\ts{t}}} \leq \norm*{\vect{g}_k\ts{t}}$  where the equality occurs if and only if $\vect{a}_k\ts*{t} = \zerovec{s}$.
This means that the prediction residual, for each agent, will never be larger (in magnitude) than the original gradient.

\subsection{Event-Triggered Residual Transmission}\label{sec:modelling g hat}

In addition to prediction-based compression, we further improve communication efficiency by avoiding transmitting residuals when the prediction is sufficiently accurate. This is in principle similar to ``event-triggered communication" in the literature of distribute optimization and control \cite{chen2018lag,zhong2010distributedevent}.

We consider a threshold-based design, where the transmission of $\vect{e}_k\ts{t}$ is omitted if $\norm*{\vect{e}_k\ts{t}}\leq e_{\text{th},k}\ts{t}$, where $e_{\text{th},k}\ts{t}\geq 0$ is a possibly time-variant threshold. The reconstructed imperfect gradient can be written as
\begin{equation}   \widetilde{\vect{g}}_k\ts{t}=\hat{\vect{g}}_k\ts{t}+\CCX*{\vect{e}_k\ts{t}},
\label{eq:reconstructed-noisy-gradient}
\end{equation}
where the ``compressed residual'' is
\begin{align}\label{eq:compressed residual}
	\CCX*{\vect{e}_k\ts{t}}
	=
	\begin{cases}
		0, & \text{if}~ \norm{\vect{e}_k\ts{t}}\leq e_{\text{th},k}\ts{t},\\
		\QQ(\vect{e}_k\ts{t}), & \text{otherwise},
	\end{cases}
\end{align}
where \QQX{} is a compression operation that maps continuous-valued residual elements to discrete-valued symbols before transmitting them over digital communication links.

\begin{figure*}[htbp]
	\centering
	\scalebox{0.95}{\includegraphics{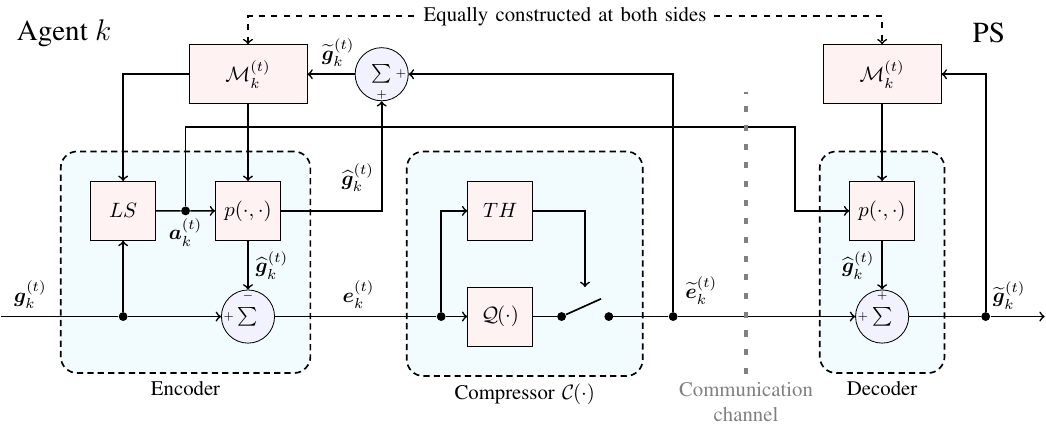}}
	\caption{Block diagram of the prediction-based gradient compression design for agent $k$. The ``compression'' block performs event-triggered transmission and classical compression. The prediction coefficients $\vect{a}_k\ts{t}$ are assumed transmitted with negligible distortion.
    }\label{fig:system block}
\end{figure*}

\begin{algorithm}[t!]
	\caption{Gradient prediction and thresholding for compression.}\label{alg:ls system}
	\DontPrintSemicolon
	\SetKwProg{MyBlock}{}{:}{}
	\SetInd{.3em}{0.4em}
	\For{$t=1,...,T$}
	{
		\ForPar{\textnormal{agents} $k=1,...,N $}
		{
			\nonl
			\MyBlock{Procedure at the agent $k$}
			{
				Compute local gradient $ \vect{g}_k\ts{t} = \nabla f_k(\vect{x}\ts{t}) $\;
				\tcp*[l]{Obtain prediction coefficients}
                Evaluate $\vect{a}^*_k\ts*{t}$ using \eqref{eq:LSE solution:closed form}\;
				\tcp*[l]{Calculate linear predictor}
                Compute $\hat{\vect{g}}_k\ts{t}$ using $\vect{a}^*_k\ts*{t}$ in \eqref{eq:linear predictor}\;
				$\vect{e}_k\ts{t} \gets \vect{g}_k\ts{t} - \hat{\vect{g}}_k\ts{t}$\;
				\If{$\norm*{\vect{e}_k\ts{t}} > e_{\textnormal{th},k}\ts{t}$}
				{
					 $\widetilde{\vect{e}}_k\ts{t} \gets \CC(\vect{e}_k\ts{t})$\;
					Transmit $ \vect{a}_k^*\ts*{t}$ and $\widetilde{\vect{e}}_k\ts{t}$ to \ac{ps}\;
					\tcp*[l]{Obtain reconstructed gradient}
                    $\widetilde{\vect{g}}_k\ts{t} \gets \hat{\vect{g}}_k\ts{t} + \widetilde{\vect{e}}_k\ts{t}$\;
				}
				\Else
				{
					Transmit $\vect{a}_k^*\ts*{t}$ to \ac{ps}\;
                    $\widetilde{\vect{g}}_k\ts{t} \gets \hat{\vect{g}}_k\ts{t}$\;
				}
                Update agent memory using $\widetilde{\vect{g}}_k\ts{t}$\;
			}
			\nonl
			\MyBlock{Procedure at the \ac{ps}}
			{
				Receive $\vect{a}_k^*\ts*{t}$, and possibly $\widetilde{\vect{e}}_k\ts{t}$\;
				Compute $\hat{\vect{g}}_k\ts{t}$ using $\vect{a}_k^*\ts*{t}$ in \eqref{eq:linear predictor}\;

				\If{ $\widetilde{\vect{e}}_k\ts{t}$ is received from agent $k$}
				{
					\tcp*[l]{Obtain reconstructed gradient}
					$\widetilde{\vect{g}}_k\ts{t} \gets \hat{\vect{g}}_k\ts{t} + \widetilde{\vect{e}}_k\ts{t}$\;
				}
				\Else
				{
                    $\widetilde{\vect{g}}_k\ts{t} \gets \hat{\vect{g}}_k\ts{t}$\;
				}
                Update \ac{ps} memory for agent $k$ using $\widetilde{\vect{g}}_k\ts{t}  $\;
			}
		}
		$ \vect{x}\ts{t+1} \gets \vect{x}\ts{t} - \gamma\sum_{k=1}^N\widetilde{\vect{g}}_k\ts{t} $\;
	}
\end{algorithm}

The complete description of our design is in \cref{alg:ls system}. A block diagram is provided in \cref{fig:system block}.
This prediction-based approach relies on the synchronized memory updating procedure to reproduce the same prediction on both sides. While in this paper we consider linear combination of past imperfect gradients, the system can use features extracted from past gradient information in an arbitrary way.

\subsection{Comparison to State-of-the-Art Methods}\label{sec:special cases}

Several existing works have considered using the past gradient(s) to make a prediction of the current gradient to reduce the communication cost.
Simple predictor designs are considered in \cite{mishchenko2019graddiff, richtarik2021ef21, chen2018lag}, where $\hat{\vect{g}}_k\ts{t}=\widetilde{\vect{g}}_k\ts{t-1}$ with different compression strategies. Using specially constructed memory updates, such as
$\smemory_k\ts{t+1} = \smemory_k\ts{t} + \beta(\widetilde{\vect{g}}_k\ts{t} - \smemory_k\ts{t})$,
we can implement momentum-style predictors \cite{chen2022sparsegraddiff}.

These aforementioned cases give no guarantee on the optimality of the predictor in any sense, and the predictor design is agnostic to the current value of the gradient. Using our method, we can scale any alternative predictor to ensure \ac{ls}-optimality. The associated communication cost is to transmit one additional scalar coefficient $a$. For example, given a predictor $\hat{\vect{g}}\ts{t}_\text{other}$, we can acquire an \ac{ls}-optimal prediction by $\hat{\vect{g}}_k\ts{t} = a \hat{\vect{g}}\ts{t}_\text{other}$ and finding the optimal $a$.

Another closely related work \cite{yue2022predictive} considers predicting the model parameters instead of the gradients. The prediction is obtained by using a set of possible prediction functions and choose the best result in every iteration.
Interestingly, despite using more sophisticated predictors, the most frequently used scheme in the numerical experiments is the model parameter difference, which is similar to the gradient difference approach but with multiple local SGD steps.

\begin{remark}
In this work, we consider using a single vector of prediction coefficients for all $d$ elements of the gradient. Our design can be generalized to the case where different prediction coefficients are used for different subsets of elements in $\vect{g}_k\ts{t}$.
This can be useful to capture a potential block structure in the gradient coefficients (\textit{e.g.} layers in a neural network).
\end{remark}

\section{Convergence Analysis of the Proposed Algorithm}
In this section, we analyze the convergence of the approach described in \cref{alg:ls system}.

{First, we rewrite \eqref{eq:reconstructed-noisy-gradient} as
\begin{align}
    &\tilde{\vect{g}}_k\ts{t} = \vect{g}_k\ts{t} -\vect{\delta}_k\ts{t},
\end{align}
where $\vect{\delta}_k\ts{t}\in\reals^d$ denotes the reconstructed gradient error (or noise), given as
\begin{align}
    \vect{\delta}_k\ts{t}=
	\begin{cases}
		{\vect{e}_k\ts{t}}, & \text{if}~ \norm*{\vect{e}_k\ts{t}}\leq e_{\text{th},k}\ts{t},\\
		\vect{e}_k\ts{t}-\QQ({\vect{e}_k\ts{t}}), & \text{otherwise}.
	\end{cases}
\end{align}
Then, the aggregated imperfect gradient can be written as
\begin{align}
    \tilde{\vect{g}}\ts{t} = \vect{g}\ts{t} - \vect{\delta}\ts{t},
\end{align}
where $\vect{\delta}\ts{t} = \sum_{k=1}^K \vect{\delta}_k\ts{t}$. We assume that $\QQ$ is an Unbiased Random Compressor according to \cref{def:unbiased compression operator}.}
\begin{definition}[Unbiased Random Compressor]\label{def:unbiased compression operator}
	A (possibly randomized) compressor $\QQ:\reals^d \rightarrow \reals^d$ is an unbiased compression operator if
	\begin{align}
		\EX{\QQ(\vect{x})} = \vect{x},
	\end{align}
	and there exists a constant $\alpha\geq 1$ such that
	\begin{align}\label{eq:urc:second moment}
		\EX{\norm{\QQ(\vect{x})}^2} \leq \alpha \norm{\vect{x}}^2.
	\end{align}
\end{definition}
\subsection{First and Second Moment Limits}
In this section, we provide some bounds on the first moment and the variance of the aggregated imperfect gradient, which will be used later for the convergence analysis.
\begin{assumption}[Bounded Threshold]\label{as:bounded threshold}
    We assume there exists a coefficient $b\ts{t} \in [0~1)$ such that
    \begin{align}\label{eq:as:bounded threshold}
        e_{\textnormal{th}}\ts{t} = \sum_{k=1}^K e_{\textnormal{th},k}\ts{t} = b\ts{t} \norm{\vect{g}\ts{t}}.
    \end{align}
    This condition \eqref{eq:as:bounded threshold} can always be met, provided that $e_{\textnormal{th},k}\ts{t}$ is adjustable over time to ensure that it remains sufficiently small relative to the gradient norm.
\end{assumption}
While Assumption \ref{as:bounded threshold} restricts the choices of $\{e_{\text{th},k}\ts{t}\}_{k\in[K]}$, it matches the intuition that the transmission threshold should decrease when the system approaches convergence as the norm of the gradient decreases.

\begin{lemma}[First Moment Bound]\label{lemma:first moment}
Given \cref{as:bounded threshold}, we can bound the first moment of $\tilde{\vect{g}}\ts{t}$ by
\begin{align}
    (1-b\ts{t}) \norm*{\vect{g}\ts{t}}^2 \leq \scalarprod{\vect{g}\ts{t}}{\EX{\tilde{\vect{g}}\ts{t}}} \leq(1+b\ts{t})\norm*{\vect{g}\ts{t}}^2.
\end{align}
\end{lemma}
\begin{proof}
See \cref{apx:first moment:new}.
\end{proof}

\begin{assumption}[Bounded Gradient Dissimilarity \cite{koloskova2020unified}]\label{as:gradient dissimilarity}
    We assume there exists some finite constants $G\geq0$ and $B\geq1$, such that in every iteration $t$, we have:
    \begin{align}
        \frac{1}{K}\sum_{k=1}^K \norm{\vect{g}\ts{t}_k}^2\leq G^2 + B^2\norm{\vect{g}\ts{t}}^2.
    \end{align}
\end{assumption}
\Cref{as:gradient dissimilarity} provides a way to bound the effects of local residual compression on the aggregated gradient norm. This assumption is less restrictive than the common assumption of uniformly bounded gradient norm \cite{alistarh2017qsgd}.

Next, we provide a bound on the variance of the aggregated imperfect gradient.
We define a coefficient $c_k\ts{t}\geq 0$,
such that the transmission threshold for each agent $k$ can be written as
\begin{align}
    e_{\text{th},k}\ts{t} = c_k\ts{t} \norm{\vect{g}_k\ts{t}}.
\end{align}
Note that $\{c_k\ts{t}\}_{k\in[K]}$ are design parameters that can be adjusted. A higher communication-triggering threshold reduces the transmission frequency of the prediction residual.
\begin{lemma}[Variance Bound]\label{lemma:variance bound}
  Let $c\ts{t} = \max\set*{c_k\ts{t}}_{k\in[K]}$ and
    $\bar{\alpha}\ts{t} = \max\set*{\alpha_k\ts{t}}_{k\in[K]}$ with
    $$\alpha_k\ts{t}= 1  + (\alpha-1) \frac{\norm{\vect{e}_{k}\ts{t}}^2}{\norm{\vect{g}_k\ts{t}}^2}.$$
    Then, under \cref{as:bounded threshold,as:gradient dissimilarity}, the variance of $\tilde{\vect{g}}\ts{t}$ is bounded by
    \begin{align}
        & \VarX{\tilde{\vect{g}}\ts{t}} \leq  K G^2 \max\set{\bar{\alpha}\ts{t} -1 , (c\ts{t})^2} \nonumber\\ &\mspace{100mu} + K B^2\max\set{\bar{\alpha}\ts{t} -1 , (c\ts{t})^2}\norm{\vect{g}\ts{t}}^2.
    \end{align}
\end{lemma}
\begin{proof}
    See \cref{apx:variance bound:new}.
\end{proof}
\begin{remark}
    Obviously, $\bar{\alpha}_k\ts{t} \in [1,\alpha]$, as the result of $\norm*{\vect{e}_{k}\ts{t}}\leq \norm*{\vect{g}_k\ts{t}}$. When the local gradient predictions are accurate, $\bar{\alpha}\ts{t}$ is close to the lower bound $1$, whereas a large prediction residual makes it closer to $\alpha$.
\end{remark}

\begin{remark}
    For each agent, there are three sources of distortion in the reconstructed gradient $\tilde{\vect{g}}_k\ts{t}$: prediction inaccuracy, compression distortion in $\QQ(\vect{e}_k\ts{t})$, and event-triggered transmission.
    For a given value of  $\alpha_k^{(t)}$ imposed by the compressor noise variance and the prediction accuracy, any $c_k\ts{t} \leq \sqrt{\alpha_k^{(t)}-1}$ will not increase the upper bound on the variance of the reconstructed gradient $\tilde{\vect{g}}_k\ts{t}$.
\end{remark}

\subsection{Convergence Analysis}
We now present the convergence analysis for our proposed algorithm with quantized prediction residuals. Before, we introduce \cref{as:L smoo,as:mu convex}.
\begin{assumption}[name=Lipschitz continuous gradients]\label{as:L smoo}
The local loss function $f_k$ has $L_k$-Lipschitz continuous gradient, \textit{i.e.}, there exists an $L_k>0$ s.t. $ \forall \vect{x},\vect{y}\in\reals^d$
	\begin{align}
		\norm{\nabla f_k(\vect{x}) - \nabla f_k(\vect{y})} \leq L_k\norm{\vect{y}-\vect{x}}^2.
		\label{eq:lipschitz-gradient}
	\end{align}
\end{assumption}

\begin{assumption}[Strong convexity]\label{as:mu convex}
	The local loss function $f_k$ is $\mu_k$-strongly convex, \textit{i.e.}, there exists an $\mu_k>0 ~\text{s.t. } \forall \vect{x},\vect{y}\in\reals^d$
	\begin{align}
		f_k(\vect{y})\geq f_k(\vect{x})+\nabla f_k(\vect{x})\transpose (\vect{y}-\vect{x})+\frac{\mu_k}{2}\norm{\vect{y}-\vect{x}}^2.
	\end{align}
\end{assumption}
\cref{as:L smoo,as:mu convex} imply that the global loss function $f(\vect{x})$ has $L$-Lipschitz continuous gradients and is $\mu$-strongly convex, with $L = \sum_{k=1}^KL_k$ and $\mu = \sum_{k=1}^K \mu_k$.

\begin{theorem}\label{thm:convergence}
Let $f^* = \min_{\vect{x}} f(\vect{x})$. Based on \cref{as:L smoo,as:mu convex,as:gradient dissimilarity,as:bounded threshold}, the iterates $\set{\vect{x}\ts{t}}_{t\in\naturals}$ generated by \cref{alg:ls system} satisfy
	\begin{align}
		&\EX{f(\vect{x}\ts{t}) -f^*}\leq \frac{\gamma L KG^2 P}{2\mu (1-b)}\\
        &~~  +(1-\mu\gamma(1-b))^{t-1}\left(f(\vect{x}\ts{1}) -f^* -\frac{\gamma L KG^2 P}{2\mu(1-b)}\right)\nonumber
	\end{align}
     if
\begin{align}
    &0 < \gamma \leq (1-b)/(L \cdot (KB^2P + (1+b)^2)),
\end{align}
where $P = \max\set{\bar{\alpha} -1 , c^2} $, $c = \max\left\{ c\ts{\tau}\right\}_{\tau\in[t-1]}$ and $b\geq\max\{ b\ts{\tau}\}_{\tau\in[t-1]}$, and
    $\bar{\alpha} \geq \max\set{\bar{\alpha}\ts{\tau}}_{\tau\in[t-1]}$.
\end{theorem}

\begin{proof}
  The proof follows directly from Theorem 4.6 in \cite{bottou2018optimization}, which requires that the first moment and the variance of the noisy gradient are bounded by some linear functions of the true gradient norm. Given the bounds presented in Lemmas \ref{lemma:first moment} and \ref{lemma:variance bound}, we can obtain the result in our theorem.
\end{proof}
\begin{remark}
    When disabling both prediction (\textit{i.e.}, $\hat{\vect{g}}_k\ts{t} = \zerovec{d}, \forall k\in[K]$) and communication threshold (\textit{i.e.}, $c=0$, then $b = 0$),
    our system reduces to \ac{d-gd} with compressed gradients satisfying \cref{as:gradient dissimilarity}.
\end{remark}

\pgfplotsset{qgd/.append style={red, loosely dashdotdotted}}
\pgfplotsset{ef21/.append style={orange, dashdotted}}
\pgfplotsset{laq/.append style={violet, densely dotted}}
\pgfplotsset{proposed1/.append style={black, solid}}
\pgfplotsset{proposed3/.append style={blue, densely dashed }}
\pgfplotsset{proposed5/.append style={cyan,  dashdotted}}

\section{Simulation Results}
\begin{figure}[tb]
	\centering
    \includegraphics{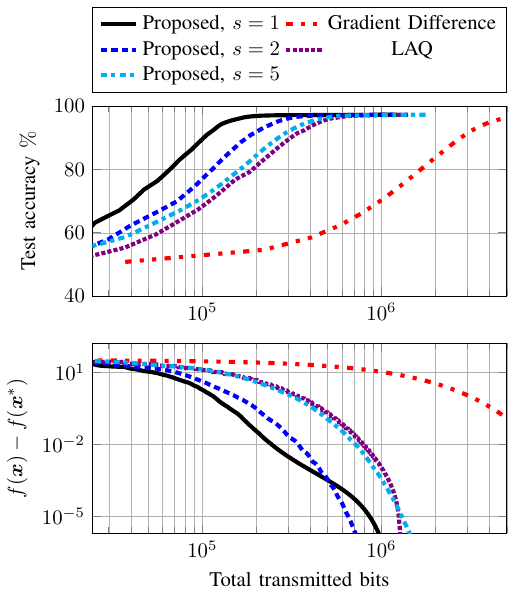}
	\caption{Training performance, using on average $R=6$ bits per residual element transmission. }
    \label{fig:res:high}
\end{figure}

We implement \cref{alg:ls system} for a regularized logistic regression task. The local loss function for agent $k$ is
\begin{align}
    f_k(\vect{x}) = \frac{1}{K} \sum_{i=1}^{N_k} \log\left( 1 + \exp\left( -y_{k,i} \vect{u}_{k,i}\transpose \vect{x} \right)\right) + \lambda \norm[2]{\vect{x}}^2,\nonumber
\end{align}
where $N_k$ is the number of data points at agent $k$, $(y_{k,i}, \vect{u}_{k,i})$ is the $i$-th data point at agent $k$,  and $\lambda =0.01$ is a regularization parameter.
We use the w8a dataset \cite{Chang2011libsvm}, which has $d=300$ parameters and $49,749$ data samples. The entire dataset is uniformly distributed across $K=10$ agents, such that $N_k = \lfloor 49749/K\rfloor = 4974$.
We use learning rate $\gamma = 0.05$.
The prediction is obtained by using $s$ past imperfect gradients. We apply $c\ts{t}=\max\set{0,(1-t/1000)/K}$, i.e., the communication threshold decreases over time until reaching $0$.

For residual (or gradient) compression, we employ a stochastic quantizer as defined in \cref{def:fixed step quantizer}.
\begin{definition}[Stochastic Fixed Interval Quantizer]\label{def:fixed step quantizer}
    Let $\Delta\in\reals_+$ denote the quantization interval. The quantizer $\QQX*{\vect{e}\ts{t}_k}$  for the $i$-th element is defined as
    \begin{align}
        [\QQX*{\vect{e}\ts{t}_k}]_i = \Delta(q_i + \xi_i)
    \end{align}
    for some $q_i\in\integers$ such that $\Delta q_i\leq [\vect{e}\ts{t}_k]_i <\Delta(q_i+1)$, where $\xi_i \sim \bernoulli{[\vect{e}\ts{t}_k]_i/\Delta -q_i}$.
\end{definition}
This unbiased quantizer maps a point to the neighboring quantization levels based on the proximity of the levels in a probabilistic manner.
After quantization, we then employ an entropy encoder (\textit{e.g.}, Huffman coding) to further compress the quantized residual.
For data transmission over digital links, we consider a rate-constrained scenario where each transmission is limited to $R\times d$ bits per iteration. This means that on average each quantized residual element consumes $R$ bits.
In addition to the compressed residuals, in every iteration each agent sends $s$ predictor coefficients and the quantization interval to the \ac{ps} with high resolution.

For performance comparison, we also simulate the following baseline schemes:
\begin{itemize}

    \item \textbf{Gradient Difference} \cite{alistarh2017qsgd}. In this case, the difference between the previously transmitted gradient and the current one is compressed and transmitted.

    \item \textbf{Lazily aggregated quantized gradient (LAQ)} \cite{sun2022laq}. LAQ uses a threshold-based criterion and activates a gradient transmission only when there is sufficient change in current gradient compared to the most recently transmitted one.
    The threshold depends on the quantization noise and a linear combination of the past $D=10$ model updates multiplied by a factor $0.8/D$.
    When a transmission occurs, gradient difference (residual) is quantized, compressed.
    Otherwise, the \ac{ps} uses the gradient from the previous iteration. If no residual is transmitted for $\bar{t} = 50$ consecutive iterations, the residual is sent regardless.
\end{itemize}
 To maintain fairness in the comparison, all schemes use the same stochastic quantizer presented in \cref{def:fixed step quantizer} and an entropy encoder.

\begin{figure}[tb]
	\centering
    \includegraphics{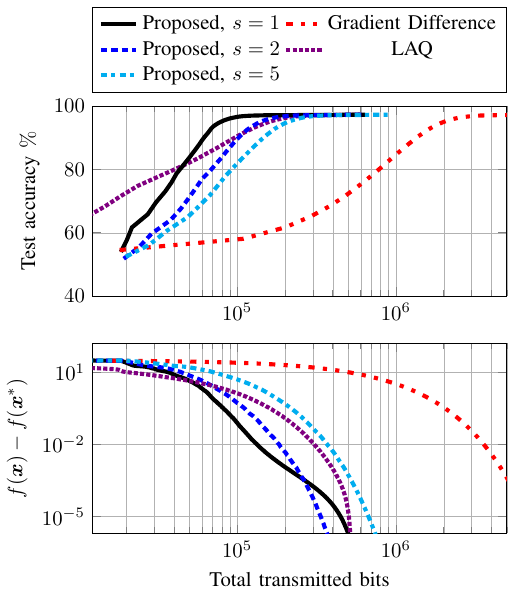}
	\caption{ Training performance, using on average $R=3$ bits per residual element transmission.}
    \label{fig:res:low}
\end{figure}

\colorlet{tabledark}{gray!30}
\colorlet{tablemid}{gray!15}
\newcommand{\mkdrk}{\cellcolor{tabledark}}
\newcommand{\mkmid}{\cellcolor{tablemid}}

\begin{table}[t!]
    \centering
    \begin{tabular}{|c|c|c|c|c|}\hline
        {\textbf{Method}} & \shortstack{\textbf{\# bits per}\\\textbf{ elem.} (R) }&{\textbf{\shortstack{Residual trans. \\frequency [$\%$]}}} & \# \textbf{iters. }& \shortstack{\# \textbf{bits}\\ $[\times 10^5]$} \\\hline
        \multirow{2}{*}{\shortstack{Gradient \\ Difference}}
        &  3.0 &  100.00 &  $ 731$ &  $   66.63$ \\ \cline{2-5}
        & \mkdrk 6.0 & \mkdrk 100 & \mkdrk $ 731$ & \mkdrk $  135.20$ \\ \hline
        \multirow{2}{*}{LAQ}
        &  3.0 &    8.18 &  $ 678$ &  $    5.06$ \\ \cline{2-5}
        & \mkdrk 6.0 & \mkdrk   9.72 & \mkdrk $ 691$ & \mkdrk $   12.49$ \\ \hline
        \multirow{2}{*}{\shortstack{ Proposed \\[-1mm] $s=1$ }}
        &  3.0 &    5.01 &  $ 725$ &  $    4.49$ \\ \cline{2-5}
        & \mkdrk 6.0 & \mkdrk   4.75 & \mkdrk $ 726$ & \mkdrk $    8.60$ \\ \hline
        \multirow{2}{*}{\shortstack{ Proposed \\[-1mm] $s=2$ }}
        &  3.0 &    1.52 &  $ 732$ &  $    3.37$ \\ \cline{2-5}
        & \mkdrk 6.0 & \mkdrk   1.37 & \mkdrk $ 733$ & \mkdrk $    6.54$ \\ \hline
        \multirow{2}{*}{\shortstack{ Proposed \\[-1mm] $s=3$ }}
        &  3.0 &    1.30 &  $ 732$ &  $    4.39$ \\ \cline{2-5}
        & \mkdrk 6.0 & \mkdrk   1.01 & \mkdrk $ 733$ & \mkdrk $    8.40$ \\ \hline
        \multirow{2}{*}{\shortstack{ Proposed \\[-1mm] $s=5$ }}
        &  3.0 &    1.28 &  $ 733$ &  $    6.73$ \\ \cline{2-5}
        & \mkdrk 6.0 & \mkdrk   0.94 & \mkdrk $ 733$ & \mkdrk $   13.00$ \\ \hline
    \end{tabular}
    \caption{Comparison of communication cost until $f(\vect{g}\ts{t}) -f(\vect{x}^*) \leq 10^{-5}$. All schemes use a stochastic quantizer followed by an entropy encoder with the same number of bits ($R\times d$) per transmission. }
    \label{tab:res:transmission prob}
\end{table}

First, we consider a high rate scenario with (on average) $R=6$ bits per element transmission, and $32$ bits for the quantization interval and predictor coefficient. \cref{fig:res:high} shows the learning performance as a function of the total number of transmitted bits. We observe notable performance gains for our proposed design as compared to the baseline schemes.
Next, we consider a low rate scenario which on average uses $R=3$ bits per element transmission, and $16$ bits per quantization interval and predictor coefficient. The results are shown in \cref{fig:res:low}.
Our design still outperforms the others, but the performance difference compared to LAQ becomes smaller.

Since our design combines prediction-based compression with an event-triggered mechanism, it is important to see how much communication is saved by omitting additional residual transmission when the prediction is sufficiently accurate. In \cref{tab:res:transmission prob}, we compare our proposed method and LAQ (both using event-triggered transmission), where we see that our method significantly reduces the residual transmission frequency. Additionally, as expected, we see that using a larger $s$ lowers the frequency of residual transmission, owing to more accurate prediction with longer memory. However, using a larger $s$ requires transmitting more predictor coefficients with high precision, which increases the communication cost.

Finally, we aim to examine the interplay between compression distortion and memory length in terms of their impacts on the communication cost. Another baseline scheme is considered for comparison, namely \textbf{EF21} \cite{richtarik2021ef21}, which is similar to Gradient Difference in terms of residual computation, but uses a Top-$L$ sparsifier, \textit{i.e.}, preserving the $L$ largest elements by magnitude.\footnote{Sparsification generally introduces bias in the received imperfect gradient, which degrades the prediction accuracy.}
Larger compression distortion causes higher noise in the imperfect gradients stored in the memory of the \ac{ps} and agents, which leads to larger prediction errors even when the local gradients are strongly correlated over time. Generally, using a longer memory improves prediction performance. However, with biased compression such as Top-$L$ sparsification, using longer memory does not necessarily reduces the communication cost, as shown in \cref{tab:res:topl}.
We see that when $L=1$, we always transmit the residual, regardless of the memory length $s$. When $L$ grows, the residual transmission frequency decreases slowly. Finding an optimal combination of memory length and compression budget remains to be further investigated.

\begin{table}[t!]
    \centering
    \begin{tabular}{|c|c|c|c|c|}\hline
        {\textbf{Method}} & $L$ & {\textbf{\shortstack{Residual trans. \\frequency [$\%$]}}} & \# \textbf{iters.}& \shortstack{\# \textbf{channel}\\\textbf{uses} $[\times10^3]$ } \\\hline
         \multirow{3}{*}{\shortstack{ EF21 }}
         & 1 & 100 & 1255 & 12.6  \\  \cline{2-5}
         &\mkmid 5 &\mkmid 100 &\mkmid 359 &\mkmid 17.9  \\  \cline{2-5}
         & \mkdrk15 &\mkdrk 100 &\mkdrk 393 & \mkdrk 58.9  \\ \hline
        \multirow{3}{*}{\shortstack{ Proposed \\[-1mm] $s=1$ }}
        & 1 & 100 & 919 & 18.3  \\  \cline{2-5}
        &\mkmid 5 &\mkmid 65.0 &\mkmid 720 &\mkmid 30.0  \\  \cline{2-5}
        &\mkdrk 15 &\mkdrk 40.0 &\mkdrk 714 & \mkdrk 51.2  \\ \hline
        \multirow{3}{*}{\shortstack{ Proposed \\[-1mm] $s=5$ }}
        & 1 & 100 & 826 & 49.6  \\  \cline{2-5}
        &\mkmid 5 &\mkmid 59.8 &\mkmid 728 &\mkmid 58.2  \\  \cline{2-5}
        &\mkdrk 15 &\mkdrk 20.6 &\mkdrk 734 & \mkdrk 59.4  \\ \hline
    \end{tabular}
    \caption{Comparison of communication cost until $f(\vect{g}\ts{t}) -f(\vect{x}^*) \leq 10^{-5}$, when using a Top-$L$ sparsifier without further quantization and entropy coding. ``\# channel uses'' counts the total number of transmissions, including residual elements and prediction coefficients (if any), sent over the channel.}
    \label{tab:res:topl}
\end{table}

\section{Conclusions}
In this work, we propose a communication-efficient gradient compression framework for distributed learning systems. After each round of local training, each agent computes a predicted gradient using past gradient information, and the prediction residual is obtained as the difference between the true and predicted gradients. The \ac{ps} performs the same prediction of the local gradients using the same memory information and optimized coefficients. Consequently, communication can be entirely avoided in a round if the prediction is sufficiently accurate. This combination of predictive encoding and event-triggered transmission leads to two types of compression noise: when the residual is transmitted or when communication is omitted. We provide both theoretical analysis of the convergence performance of the proposed design and simulation results showing its effectiveness in improving training performance with a reduced number of transmitted bits. Future directions include exploring the potential correlation among agents to perform joint prediction and user clustering.

\bibliographystyle{IEEEtran}
\bibliography{latex-setup/otherfull,IEEEfull,latex-setup/ref}

\appendix

\subsection{Proof of Lemma \ref{lemma:first moment} (First Moment Bound)}\label{apx:first moment:new}
\begin{align}\label{eq:apx:first:scalar}
    \scalarprod{\vect{g}\ts{t}}{\EX{\tilde{\vect{g}}\ts{t}}} & = \scalarprod{\vect{g}\ts{t}}{\vect{g}\ts{t} + \EX{\vect{\delta}\ts{t}}}\nonumber\\
    &= \norm{\vect{g}\ts{t}}^2+\scalarprod{\vect{g}\ts{t}}{\EX{\vect{\delta}\ts{t}}}
\end{align}

\begin{itemize}
    \item
When $\norm*{\vect{e}_k\ts{t}} > e_\text{th,k}\ts{t}$, the compressed residual $\tilde{e}_k\ts{t}$ is transmitted. Using the fact that $\QQX{\cdot}$ is unbiased, $\EX{\vect{\delta}\ts{t}} = \vect{e}\ts{t} - \vect{e}\ts{t} = 0$.
\item
When $\norm*{\vect{e}_k\ts{t}} \leq e_\text{th,k}\ts{t}$, the residual transmission is omitted, and a bias is introduced. This bias will be largest when all agent simultaneously omit transmission. Consequently, we can lower bound \eqref{eq:apx:first:scalar} by
\begin{align}
    \scalarprod{\vect{g}\ts{t}}{{\tilde{\vect{g}}\ts{t}}} &\geq \norm{\vect{g}\ts{t}}^2-\norm{\vect{g}\ts{t}}\norm{\vect{\delta}\ts{t}}\nonumber\\
    &\geq \norm{\vect{g}\ts{t}}^2-\norm{\vect{g}\ts{t}} \left(\sum_{k=1}^K e_{\text{th},k}\ts{t}\right)\nonumber\\
    & = \norm{\vect{g}\ts{t}}^2-\norm{\vect{g}\ts{t}} e_{\text{th}}\ts{t}.
\end{align}
From \cref{as:bounded threshold}, one obtains
\begin{align}\label{eq:apx:first:lower}
    \scalarprod{\vect{g}\ts{t}}{{\tilde{\vect{g}}\ts{t}}} &= \norm{\vect{g}\ts{t}}^2(1-b\ts{t}) > 0,
\end{align}

By upper bounding the scalar product in \eqref{eq:apx:first:scalar}, we get similar calculations for the upper bound
    \begin{align}
    \scalarprod{\vect{g}\ts{t}}{\tilde{\vect{g}}\ts{t}}
    &\leq \norm{\vect{g}\ts{t}}^2+\norm{\vect{g}\ts{t}}\norm{\vect{\delta}\ts{t}}\nonumber\\
    &\leq \norm{\vect{g}\ts{t}}^2+\norm{\vect{g}\ts{t}} \left(\sum_{k=1}^K e_{\text{th},k}\ts{t}\right)\nonumber\\
    &= \norm{\vect{g}\ts{t}}^2+\norm{\vect{g}\ts{t}} e_{\text{th}}\ts{t}\nonumber\\
    &= \norm{\vect{g}\ts{t}}^2(1+b\ts{t}).\label{eq:apx:first:upper}
\end{align}
\end{itemize}
Combining \eqref{eq:apx:first:lower} and \eqref{eq:apx:first:upper}, we get the desired bound
\begin{align}
    (1-b\ts{t}) \norm{\vect{g}\ts{t}}^2 \leq \scalarprod{\vect{g}\ts{t}}{\EX{\tilde{\vect{g}}\ts{t}}} \leq(1+b\ts{t})\norm{\vect{g}\ts{t}}^2.
\end{align}

\subsection{Proof of Lemma \ref{lemma:variance bound} (Variance Bound)}\label{apx:variance bound:new}
As each agent compresses independently, the variance can be separated in the sum of local variances as
\begin{align}\label{eq:apx:var}
    \VarX{\tilde{\vect{g}}\ts{t}} = \sum_{k=1}^K \VarX{\tilde{\vect{g}}_k\ts{t}} = \sum_{k=1}^K \VarX{{\vect{\delta}}_k\ts{t}}.
\end{align}
\begin{itemize}
    \item
If agent $k$ does not transmit the residual,
$\norm*{\vect{\delta}_k\ts{t}} \leq e_{\text{th},k}\ts{t} = c_k\ts{t}\norm*{\vect{g}_k\ts{t}}\leq c\ts{t}\norm*{\vect{g}_k\ts{t}} $, and
\begin{align}\label{eq:apx:var:single agent:th}
    \VarX{{\vect{\delta}}_k\ts{t}} \leq (e_{\text{th},k}\ts{t})^2 = (c_k\ts{t})^2\norm{\vect{g}_k\ts{t}}^2 \leq (c\ts{t})^2\norm{\vect{g}_k\ts{t}}^2.
\end{align}

\item
If agent $k$ transmits its residual, \textit{i.e.}, $\norm*{\vect{e}_k\ts{t}} > e_\text{th,k}\ts{t}$, $\vect{\delta}_k\ts{t}$ has zero mean, and has the variance
\begin{align}
    \VarX{{\vect{\delta}}_k\ts{t}} &=
     \EX{\norm{ \vect{e}_k\ts{t} - \EX*{\QQX*{\vect{e}_k\ts{t}}} }^2} \nonumber \\
	&= \EX{\norm{\vect{e}_k\ts{t}}^2} - 2\EX{\scalarprod{\vect{e}_k\ts{t}}{\QQX*{\vect{e}_k\ts{t}}}} \nonumber\\&\mspace{160mu}+\EX{\norm{\QQ(\vect{e}_{k}\ts{t})}^2} \nonumber \\
    &= \EX{\norm{\QQ(\vect{e}_{k}\ts{t})}^2} - \norm{\vect{e}_k\ts{t}}^2.
\intertext{Using \eqref{eq:urc:second moment}, one gets}
	\VarX{{\vect{\delta}}_k\ts{t}}&\leq \alpha\norm{\vect{e}_{k}\ts{t}}^2 - \norm{\vect{e}_k\ts{t}}^2 \nonumber \\
    &= (\alpha-1)\norm{\vect{e}_{k}\ts{t}}^2
	= (\alpha_{k}\ts{t}-1)\norm{\vect{g}_{k}\ts{t}}^2 \label{eq:apx:var:single agent:alpha}
\end{align}
where
\begin{align}
    \alpha_{k}\ts{t} & = 1  + (\alpha-1) \frac{\norm{\vect{e}_{k}\ts{t}}^2}{\norm{\vect{g}_k\ts{t}}^2}.
    \label{eq:comp pred:alpha improvement residual}
\end{align}
Since $\norm*{\vect{e}_{k}\ts{t}} \leq \norm*{\vect{g}_{k}\ts{t}}$, $\alpha_{k}\ts{t} \in [1,\alpha]$.
Let $\bar{\alpha}\ts{t} = \max\set*{ \alpha_{k}\ts{t}}_{k\in[K]}$ and insert \eqref{eq:apx:var:single agent:th} and \eqref{eq:apx:var:single agent:alpha} into \eqref{eq:apx:var}, we get
\begin{align}\label{eq:apx:final variance}
    \VarX{\tilde{\vect{g}}\ts{t}}  \leq& \sum_{k=1}^K \max\set{\alpha_k\ts{t} -1 , (c_k\ts{t})^2}\norm{\vect{g}_k\ts{t}}^2 \nonumber\\
     \leq& \max\set{\bar{\alpha}\ts{t} -1 , (c\ts{t})^2} \sum_{k=1}^K \norm{\vect{g}_k\ts{t}}^2. \nonumber\\
 \intertext{Using \cref{as:gradient dissimilarity}, we get}
    \VarX{\tilde{\vect{g}}\ts{t}} \leq& K \max\set{\bar{\alpha}\ts{t} -1 , (c\ts{t})^2}\left(G^2 + B^2\norm{\vect{g}\ts{t}}^2\right)\nonumber\\
    =& K G^2 \max\set{\bar{\alpha}\ts{t} -1 , (c\ts{t})^2} \\ & ~~ +  K B^2\max\set{\bar{\alpha}\ts{t} -1 , (c\ts{t})^2}\norm{\vect{g}\ts{t}}^2.\nonumber
\end{align}

\end{itemize}
\end{document}